\documentclass[a4paper,twocolumn,11pt,accepted=2017-05-09]{quantumarticle}
\pdfoutput=1
\usepackage[utf8]{inputenc}
\usepackage[english]{babel}
\usepackage[T1]{fontenc}
\usepackage{amsmath}
\usepackage{hyperref}
\usepackage{tikz}
\usepackage{lipsum}
\usepackage{amsmath}  
\usepackage{amssymb}  
\usepackage{amsthm}   
\usepackage{ulem}   
\newtheorem{corollary}{Corollary}
\newtheorem{remark}{Remark}
\newtheorem{theorem}{Theorem}
\newtheorem{proposition}{Proposition}
\usepackage{fancyhdr}  
\usepackage[utf8]{inputenc}
\usepackage{authblk} 

\begin{document}
	
	\title{A new fidelity of quantum channel evolution and its geometric interpretation}
	
	\author[1]{Xiaojing Yan} 
	\author[1]{Xiao Sun}   
	\author[2]{Mingming Du} 	
	\author[1]{Jiashan Tang} 
	\email{tangjs@njupt.edu.cn}

	\affil[1]{College of Science, Nanjing University of Posts and Telecommunications, 210003 Nanjing, P R China}
	\affil[2]{College of Electronic and Optical Engineering, Nanjing University of Posts and Telecommunications, 210003 Nanjing, P R China}

	
	\maketitle
	
\begin{abstract}
		Fidelity is crucial for characterizing transformations of quantum states under various quantum channels, which can be  served as a fundamental tool in resource theories. Firstly, we define an \(\alpha\)-\(z\)-fidelity as a significant quantity in quantum information theory and give the properties of the fidelity with orders \(\alpha\) and \(z\).
		Secondly, by analyzing the $\alpha$-$z$-fidelity under the evolution of different types of quantum channels (single orbit, all quantum channels, unitary quantum channels, and mixed unitary quantum channels), we propose a limit formula for the maximum and the minimum of the $\alpha$-$z$-fidelity. In addition, we have extended the $\alpha$-$z$-Rényi relative entropy, providing new insights into its relevance for resource quantification. Finally, we offer a geometric interpretation for measuring the distance between quantum states, contributing to the broader understanding of the operational and transformative power of dynamical quantum resources across various physical settings.
	\end{abstract}


	\section{Introduction}
	Fidelity was firstly proposed by Uhlmann using the concept of ``transition probability'' (\cite{AU}) which plays an important role in quantum information. Quantum state not only is carriers of information, but also enables the completion of complex computational processes through controlling within a micro-environment(\cite{BS,TKL}). Therefore, the fidelity between two quantum states $\rho$ and $\sigma$ can be served as a ``distance'' measure to elucidate two important questions in quantum information science(\cite{MANI}), namely, (i) the meaning of similarity between two pieces of information, and (ii) the meaning of an operation maintaining information. 
	\par 
	In quantum information science, fidelity has been extended to $\alpha$-fidelity
	by Dominique Spehner(\cite{SD,YXZ}), which generalizes the traditional quantum fidelity as a distance measure for quantum state similarity. When $\alpha$ = 1/2, $\alpha$-fidelity coincides with the standard quantum fidelity. This generalization, along with other measures like sub-fidelity and superfidelity, provides a multi-dimensional framework for comparing quantum states. Researchers aim to integrate these measures into a unified mathematical framework that satisfies the data processing inequality, enhancing quantum information processing efficiency. In this note, 
	motivated by two-parameter functions(\cite{GABA, KMRA}), we propose  a new  two-parameter quantum fidelity, $\alpha$-$z$-fidelity, which is denoted by $F_{\alpha, z}(\rho, \sigma)$. The new conception not only generalized the above mentioned $\alpha$-fidelity, but also provides a new framework for the research of fidelity.
	
	\par
	With the development of quantum computing and communication technologies, the study of dynamic quantum resources has become a hot topic, focusing on how to leverage quantum features such as quantum states, entanglement, and interference to enhance information processing capabilities or optimize tasks(\cite{SP}). During the evolution of a quantum channel, the dynamic changes in quantum resources have  impact directly on the measurement of fidelity, especially on the decoherence or disentanglement effects(\cite{BCH}). The $\alpha$-$z$-fidelity, as a measure of quantum state similarity, not only helps to analyze quantum state changes but also provides theoretical support for optimizing quantum resources. This paper introduces the $\alpha$-$z$-fidelity and proposes a new method to study quantum state evolution.
	
	\par 
	Currently, researchers are striving to construct a mathematical framework that integrates these measures and optimizes their properties, especially in the aspect of satisfying the data processing inequality with the hope of expanding the applicability of fidelity parameters, and thus achieving higher efficiency and accuracy in quantum information processing. Fortunately, such kind of mathematical framework does exist. For the sake, we introduce a family of two-parameter quantum fidelity, $F_{\alpha,z}(\rho,\sigma)$ for $\alpha \text{, } z \in (0, \infty)$. 
	The $\alpha$-$z$-fidelity will reduce to the  $\alpha$-fidelity if $\rho$ and $\sigma$ are 
	commuting operators, moreover, it will reduce to the classical  fidelity if $\alpha$ = $z $ = 1/2.
	
	\par It has been determined,  in the literatures  \cite{YXZ, ZLFS}, that the maximum and minimum values of the quantum $\alpha$-fidelity and Rényi relative entropy between two quantum states have been successfully applied to the field of quantum mechanics. Inspired by these works, we now turn our attention to the $\alpha$-$z$-fidelity in the presence of a unitary dynamics, as well as the distance between two quantum states in a general unitary evolution: namely, the maximum and minimum values of $\alpha$-$z$-fidelity, $F_{\alpha,z}(V\rho V^*, W\sigma W^*)$, where \(V\)and \(W\) are unitary operations. This research will further reveal the changes in quantum state accuracy in the process of evolution, presenting  new insights for quantum information science.
	\par 
	The paper is organized as follows: Sec. II defines this new family of fidelity and provides some properties of the $\alpha$-$z$-fidelity. Sec. III  presents the main results for the quantum \(\alpha\)-\(z\)-fidelity of unitary orbits: maximizing and minimizing \( F_{\alpha,z}(\rho, U\sigma U^*) \). Moreover, it is proved that, under certain conditions, the values of the quantum \(\alpha\)-\(z\)-fidelity may traverse some special interval. Sec. IV studies the minimum and maximum values of the \(\alpha\)-$z$ fidelity of the quantum state \(\rho\) under the evolution of quantum channels. It primarily considers three types of quantum channels: all quantum channels, unitary quantum channels, and mixed unitary channels. By introducing the concepts of closed density operators and the convex set of quantum channels, Sec V provides a geometric interpretation of the $\alpha$-$z$-fidelity. From a geometric perspective, the $\alpha$-$z$-fidelity not only helps to analyze the robustness of quantum information against different types of noise or disturbances, but also plays an important role in understanding the transformation of quantum resources under various quantum operations.
	
	\section{Preliminaries}
	Let \( \mathcal{H}_d \) be a finite-dimensional Hilbert space of dimension \( d \) (where \( d < \infty \)). The set of linear operators acting on \( \mathcal {H}_d \) is denoted by \( \mathcal {L(H}_d) \). Define \(\mathcal{P}(\mathcal{H}_d)\) as the set of all positive semi-definite operators acting on \(\mathcal{H}_d\), i.e., \(
	\mathcal{P}(\mathcal{H}_d) := \{X \in \mathcal{L}(\mathcal{H}_d) : X \geq 0\}.
	\)
	Similarly, define \(\mathcal{D}(\mathcal{H}_d)\) as the set of density operators on \(\mathcal{H}_d\), i.e., \(
	\mathcal{D}(\mathcal{H}_d) := \{\rho \in \mathcal{P}(\mathcal{H}_d) : \mathrm{Tr}(\rho) = 1\}.\)
	
	\par For quantum states $\rho, \sigma \in \mathcal{D}(\mathcal{H}_d)$ i.e., $\alpha, z \in (0, \infty)$, the quantum $\alpha$-$z$-fidelity between $\rho$ and $\sigma$ is defined as follows:
	\begin{equation} \label{2-1}
		\begin{aligned}
			F_{\alpha, z}(\rho, \sigma) & := \left( \text{Tr} \left[ \left( \sigma^{\frac{1-\alpha}{2z}} \rho^{\frac{\alpha}{z}} \sigma^{\frac{1-\alpha}{2z}} \right)^{z} \right] \right)^\frac{1}{\alpha} \\
			& = \left( \text{Tr} \left[ \left( \rho^{\frac{\alpha}{2z}} \sigma^{\frac{1-\alpha}{z}} \rho^{\frac{\alpha}{2z}} \right)^{z} \right] \right)^{\frac{1}{\alpha}}.
		\end{aligned}
	\end{equation}
	When \( \rho \) is non-degenerate, \( F_{\alpha, z}(\rho, \sigma) \) is continuous for \( \alpha, z \in (0, \infty) \). It is clear to see that
	\[
	F(\rho,\sigma)=F_{(\frac{1}{2},\frac{1}{2})}(\rho,\sigma), \:  F_{\alpha}(\rho,\sigma)= F_ {(\alpha,\alpha)}(\rho ,\sigma)\text{,}
	\]
	where \( F(\rho, \sigma) \) and \( F_{\alpha}(\rho, \sigma) \) are the quantum fidelity and \(\alpha\)-fidelity, respectively.\par When $\rho$ and $\sigma$ can commute, i.e. $\rho \sum_{i=1}^{d} \tilde{p}_i |i\rangle\langle i|$,  $\sigma = \sum_{i=1}^{d} \tilde{q}_i |i\rangle\langle i|$, then the quantum $\alpha$-z-fidelity simplifies to the following $\alpha$-fidelity:
	\begin{equation}\label{2-2}
		\begin{aligned}
			F_{\alpha,z}^{C}(\tilde{p},\tilde{q}) = \left(\sum_{i=1}^{d}\tilde{p}_{i}^\alpha \tilde{q}_{i}^{1-\alpha}\right)^{\frac{1}{\alpha}} = F_{\alpha}^{C}(\tilde{p},\tilde{q})\text{,}		
		\end{aligned}
	\end{equation}
	\noindent where $\tilde{p} = \{\tilde{p}_1, \ldots, \tilde{p}_d\}$ and $\tilde{q} = \{\tilde{q}_1, \ldots, \tilde{q}_d\}$ are two probability distributions.\par
	
	If $\operatorname{supp}(\rho) = \operatorname{supp}(\sigma)$, then $F_{\alpha, -z} = F_{\alpha, z}$. In other words, $F_{\alpha, z}$ is an even function with respect to   $z$(\cite{KMRA}). In this paper, to simplify the problem, we will only consider the region that \( z > 0 \) .
	
	To be a useful distance between two quantum states, one may expect that the $\alpha$-$z$-fidelity satisfies some properties, including the data processing inequality in the quantum information theoretic framework. Now, we will list some properties as follows.
	\begin{proposition}\label{pp2.1}
		Let $\rho, \sigma \in \mathcal{D(H}_d)$,
		\par 
		1. Unitary invariance(\cite{GS}).  For any unitary \( U \in \mathcal{L(H}_d) \), \( F_{\alpha,z}(U \rho U^*, U \sigma U^*) = F_{\alpha,z}(\rho, \sigma) \)
			
			2. Concavity and Convexity(\cite{ZH}). \( F_{\alpha, z}\) is concave, for \(\alpha, z \in \{(\alpha, z) \mid 0 < \alpha < 1, z \geq \max\{\alpha,1-\alpha\}\}\). \( F_{\alpha, z} \) is convex, for \( \{(\alpha, z) \mid 1 < \alpha \leq 2, \frac{\alpha}{2}\leq z \leq \alpha \}\bigcup \{(\alpha, z) \mid 2 \leq \alpha < \infty, \alpha-1 \leq z \leq \alpha\}\).
			
			3.Data processing inequality(\cite{GS,HZ})  Let \(\Phi\) be a completely positive and trace preserving map (or quantum channel) on \(\mathcal{D(H}_d)\). For \(\alpha, z \in  \{(\alpha, z) \mid 1 < \alpha \leq 2, \frac{\alpha}{2}\leq z \leq \alpha \}\bigcup \{(\alpha, z) \mid 2 \leq \alpha < \infty, \alpha-1 \leq z \leq \alpha\}\),
			\begin{equation}\label{2-3}
				F_{\alpha,z}(N(\rho),N(\sigma) \leq F_{\alpha,z}(\rho,\sigma).
			\end{equation}	
		\end{proposition}
		\begin{proposition}\label{pp2.2} Some math properties, 
			\par 
			1. Golden-Thompson inequality(\cite{DSMB}).  Let \( A, B \in \mathcal{L(H}_d) \) be two Hermitian operators on the Hilbert space \( \mathcal{H}_d \). Then,
			\begin{equation}\label{2-4}
				\begin{aligned}
					\operatorname{Tr}[\exp(A + B)] \leq \operatorname{Tr}[\exp(A) \exp(B)]\text{,}
				\end{aligned}
			\end{equation}
			holds, with equality if and only if \( A \) and \( B \) commute.
			\par 
			2. Araki-Lieb-Thirring Inequality(\cite{EC}).  Let \(A,B\in \mathcal{P(H}_d)\) and \(q > 0\),then
			\begin{equation}\label{2-5}
				\begin{aligned}
					&\operatorname{Tr}[(B^\frac{r}{2}A^rB^\frac{r}{2})^\frac{q}{r}] \leq  \operatorname{Tr}[(B^\frac{1}{2}AB^\frac{1}{2})^q], \quad 0 < r \leq 1\text{,}\\
					&\operatorname{Tr}[(B^\frac{r}{2}A^rB^\frac{r}{2})^\frac{q}{r}] \geq  \operatorname{Tr}[(B^\frac{1}{2}AB^\frac{1}{2})^q], \quad r \geq 1.
				\end{aligned}
			\end{equation}
			\par
			
			\par
			3. Schur Convexity and Concavity(\cite{PS}).  Let 
			\(
			D = \{\tilde{p} = (\tilde{p}_1, \tilde{p}_2, \ldots, \tilde{p}_d) \in \mathbb{R}^d : \tilde{p}_1 \geq \tilde{p}_2 \geq \cdots \geq \tilde{p}_d \}
			\text{ be a subset of } \mathbb{R}^d, \text{ and } \phi \text{ satisfies } \phi : D \to \mathbb{R}. 
			\)
			Define
			\begin{equation*}
				\phi(\tilde{q}) = \sum_{i=1}^{d} \nu_i g(\tilde{q}_i)\text{,}
			\end{equation*}
			where \( \nu = (\nu_1, \ldots, \nu_d) \in D \) and \( \nu_i \geq 0 \text{ for } i = 1, \ldots, d. \) If \( g(\tilde{q}_i) \) is non-decreasing, then \( \phi \) is Schur-convex, If \( g(\tilde{q}_i) \) is non-increasing, then \( \phi \) is Schur-concave. 
		\end{proposition}
		
		\begin{proposition}[\cite{ZX}]\label{pp2.3}
			For two quantum states \(\rho, \sigma \in \mathcal{D(H}_d)\) with eigenvalues
			\(
			\lambda(\rho) = (\lambda_1, \ldots, \lambda_d )  \text{ and }  \lambda(\sigma) = (\mu_1, \ldots, \mu_d),
			\) we have
			\begin{equation}\label{2-6}
				\langle \lambda(\rho), \lambda(\sigma) \rangle := \sum_{i=1}^{d} \lambda_i \mu_i
			\end{equation}
			and \begin{equation}\label{2-8}
				\langle \lambda^{\downarrow}(\rho), \lambda^{\uparrow}(\sigma) \rangle\leq \operatorname{Tr}[\rho \sigma] \leq \langle \lambda^{\downarrow}(\rho), \lambda^{\downarrow}(\sigma) \rangle,
			\end{equation}
			where \(\lambda^{\downarrow}(\rho)\) and \(\lambda^{\uparrow}(\rho)\) are the eigenvalues of \(\lambda(\rho)\) rearranged in decreasing and increasing order, respectively.
		\end{proposition}
		
		\section{Quantum $\alpha$-z-fidelity between unitary orbits }
		Let \(\mathcal{D(H}_d)\) be the set of \( d \times d \) unitary matrices acting on a \( d \)-dimensional Hilbert space \( \mathcal{H}_d \). For the quantum state \( \rho \in \mathcal{D(H}_d)\), its unitary orbit is given by(\cite{ZLFS,SJJD})
		\begin{equation*}
			U_\rho = \{ U \rho U^\ast : U \in \mathcal{(H}_d) \}.
		\end{equation*}

		The goal is to study the extremal values of \( F_{\alpha,z} \) over these unitary orbits. Specifically, for any unitary operators \( V \) and \( W \), we have,
		\begin{equation*}\label{Equation 3-1}
			\begin{aligned}
				\min F_{\alpha,z}(V \rho V^\ast, W \sigma W^\ast) &= \min F_{\alpha,z}(\rho, U \sigma U^\ast)\text{,} \\
				\max F_{\alpha,z}(V \rho V^\ast, W \sigma W^\ast) &= \max F_{\alpha,z}(\rho, U \sigma U^\ast).
			\end{aligned}
		\end{equation*}
		where \(U = V^\ast W\). Therefore, finding the minimum or maximum of \( F_{\alpha,z} \) over all unitary orbits reduces to solving \(
		\min F_{\alpha,z}(\rho, U \sigma) \text{ and }
		\max F_{\alpha,z}(\rho, U \sigma).
		\)
		\begin{theorem}\label{th3.1}
			Let \( \rho, \sigma \in \mathcal{D(H}_d) \). The quantum \(\alpha\)-z-fidelity between the unitary orbits \( U_\rho \) and \( U_\sigma \) satisfies the following relations, for \( z \in (0, \infty) \),
			\begin{equation}\label{3-9}
				\begin{aligned}
					&\max_{\substack{U \in U\mathcal{(H}_d)}}F_{\alpha,z}(\rho, U_{\sigma})\\&= \begin{cases}
						F_{\alpha}^{C}(\lambda^{\downarrow}(\rho),\lambda^{\downarrow}(\sigma)), & \alpha \in(0,1)\text{,}\\
						F_{\alpha}^{C}(\lambda^{\downarrow}(\rho),\lambda^{\uparrow}(\sigma)), & \alpha \in (1,\infty)\text{,}
					\end{cases}
				\end{aligned}
			\end{equation}
			for \(z\in (0,1) \),
			\begin{equation}\label{3-10}
				\begin{aligned}
					&\min_{\substack{U \in U\mathcal{(H}_d)}}F_{\alpha,z}(\rho, U_{\sigma})\\&= \begin{cases}
						F_{\alpha}^{C}(\lambda^{\downarrow}(\rho),\lambda^{\uparrow}(\sigma)), &  \alpha \in(0,1)\text{,} \\
						F_{\alpha}^{C}(\lambda^{\downarrow}(\rho),\lambda^{\downarrow}(\sigma)), &  \alpha \in (1,\infty)\text{,}
					\end{cases}
				\end{aligned}
			\end{equation}
			and for \(\alpha, z \in  \{(\alpha, z) \mid 1 < \alpha \leq 2, \frac{\alpha}{2}\leq z \leq \alpha \}\bigcup \{(\alpha, z) \mid 2 \leq \alpha < \infty, \alpha-1 \leq z \leq \alpha\}\),
			\begin{equation}\label{3-11}
				\begin{aligned}
					\min_{\substack{U \in U\mathcal{(H}_d)}}F_{\alpha,z}(\rho, U_{\sigma}) = F_{\alpha}^{C}(\lambda^{\downarrow}(\rho),\lambda^{\downarrow}(\sigma)).
				\end{aligned}
			\end{equation}
			where \(\lambda^{\downarrow}(\rho)\) and \(\lambda^{\downarrow}(\sigma)\)  (resp.  \(\lambda^{\uparrow}(\rho)\) and \(\lambda^{\uparrow}(\sigma)\)) are the eigenvalues of states $\rho$ and $\sigma$,  which are taken by rearranging the numbers \(\lambda(\rho)\) and \(\lambda(\sigma)\) in decreasing  order (resp. increasing order), respectively.
		\end{theorem}
		\begin{proof}
			Consider the spectral decomposition of $\rho$, $\sigma$ \(\in \mathcal{D(H}_d) \),
			\[
			\rho = \sum_{i=1}^{d}\lambda_{i}^{\downarrow}(\rho)|i\rangle\langle i| \text{ and }  \sigma = \sum_{i=1}^{d}\lambda_{i}^{\downarrow}(\sigma)W_0|i\rangle\langle i| W_{0}^*\text{,}
			\]
			where \(W_0\) is a unitary operator. Without loss of generality, we assume that \(\lambda_{i}^{\downarrow}(\rho)>0\)  and  \(\lambda_{i}^{\downarrow}(\sigma )> 0, i=1,...,d\). Otherwise, there exists an \(i\) such that \(\lambda_{i}^{\downarrow}(\rho)=0\) (resp. \(\lambda_{i}^{\downarrow}(\sigma)=0\)), considering  \(\lambda_{i}^{\downarrow}(\rho+\xi I)>0\) (resp. \(\lambda_{i}^{\downarrow}(\sigma+\xi I)>0)\).
			\\
			
			\noindent\textbf{(I)}. For Equation \eqref{3-9}, when \( z \in (0, \infty) \) and \( \alpha \in (0,1) \), taking 
			\begin{equation*}
				\begin{aligned}
					A = \frac{1 - \alpha}{z} U \ln(\sigma) U^\ast \text{ and }
					B = \frac{\alpha}{z} \ln(\rho).
				\end{aligned}
			\end{equation*}
			According to Golden-Thompson inequality, there exists a unitary operator \( U_1 \) and \( U_2 \) such that,
			\begin{equation}
				\begin{aligned}
					&\exp\left(\frac{A}{2}\right) \exp(B) \exp\left(\frac{A}{2}\right) \\&=\exp\left(U_1 A U_1^\ast + U_2 B U_2^\ast\right).
				\end{aligned}
			\end{equation}
			It follows that,
			\begin{equation*}
				\begin{aligned}
					&F_{\alpha,z}(\rho, U \sigma U^\ast) \\&= \left(\operatorname{Tr}\left(U \sigma^{\frac{1 - \alpha}{2z}} U^\ast \rho^{\frac{\alpha}{z}} U \sigma^{\frac{1 - \alpha}{2z}} U^\ast\right)^z\right)^{\frac{1}{\alpha}} \\
					&= \left( \operatorname{Tr}\left(\exp\left(\alpha U_2 \ln \rho U_2^\ast + \right. \right. \right. \\
					&\quad \left. \left. \left. (1 - \alpha) U_1 U \ln \sigma U^\ast U_1^\ast\right)\right)\right)^{\frac{1}{\alpha}}.
				\end{aligned}
			\end{equation*}
			
			Define \( \tilde{U} = U_2^\ast U_1 U \), then,
			\begin{equation*}
				\begin{aligned}
					&F_{\alpha,z}(\rho, U \sigma U^\ast) \\ &= \left(\operatorname{Tr}\left(\exp\left(\alpha \ln \rho + (1 - \alpha) \tilde{U} \ln \sigma \tilde{U}^\ast\right)\right)\right)^{\frac{1}{\alpha}}.
				\end{aligned}
			\end{equation*}
			Applying the Golden-Thompson inequality, we obtain
			\begin{equation*}
				\begin{aligned}
					F_{\alpha,z}(\rho, U \sigma U^\ast) &\leq \left(\operatorname{Tr}\left(\rho^{\alpha} \tilde{U} \rho^{1 - \alpha} \tilde{U}^\ast\right)\right)^{\frac{1}{\alpha}}.
				\end{aligned}
			\end{equation*}
			Next, using the Araki-Lieb-Thirring inequality, with \( r = q = \alpha \), \( A = \rho \), and \( B = \tilde{U} \sigma^{\frac{1 - \alpha}{\alpha}} \tilde{U}^\ast \), we derive that
			\[
			F_{\alpha,z}(\rho, U \sigma U^\ast) \leq F_{\alpha}(\rho, \tilde{U} \sigma \tilde{U}^\ast).
			\]
			The unitary group \( U\mathcal{(H}_d) \) is compact, and  when defined on the unitary group equipped with the operator norm, the map \( U \mapsto F_{\alpha,z}(\rho, U \sigma U) \) is continuous (as shown in the proof of Theorem \ref{th3.2}). We have demonstrated that this map is continuous with respect to the 1-norm. Moreover, since all norms are equivalent in finite-dimensional vector spaces, there exists a unitary operator \( U_0 \) for which the maximum value is attained. Thus, for the inequality established above, we obtain the desired result, that is, \( 0 < \alpha < 1 \),
			\begin{equation*}
				\begin{aligned}
					\max_{\substack{U \in U\mathcal{(H}_d)}} F_{\alpha,z}(\rho, U \sigma U^\ast) &= F_{\alpha,z}(\rho, U_0 \sigma U_0^\ast) \\ &= F_{\alpha}(\rho, \tilde{U}_0 \sigma \tilde{U}_0^\ast).
				\end{aligned}
			\end{equation*}
			Consequently,
			\begin{equation*}
				\begin{aligned}
					&\operatorname{Tr}[\exp(\alpha \ln \rho + (1 - \alpha) \tilde{U}_0 \ln \sigma \tilde{U}_0^\ast)] \\ &	= \operatorname{Tr}[\rho^{\alpha} \tilde{U}_0 \sigma^{1 - \alpha} \tilde{U}_0^\ast].
				\end{aligned}
			\end{equation*}
			By the equality condition of the Golden-Thompson inequality, we obtain,
			\[
			[\rho^\alpha, \tilde{U}_0 \sigma^{1 - \alpha} \tilde{U}_0^\ast] = 0.
			\]
			It follows that \(\tilde{U}_0 = W_0^\ast\), which implies
			\(
			[\rho, W_0^\ast \sigma W_0] = 0
			\).
			Thus, for \(0 < \alpha < 1\), the maximum of \(F_{\alpha,z}(\rho, U \sigma U^\ast)\) is achieved, then
			\[
			\max_{\substack{U \in U\mathcal{(H}_d)}} F_{\alpha,z}(\rho, U \sigma U^\ast) = F_{\alpha}(\rho, W_0^\ast \sigma W_0).
			\]
			Using \( [\rho^\alpha, W_0^\ast \sigma^{1 - \alpha} W_0] = 0 \), we obtain
			\[
			\max_{\substack{U \in U\mathcal{(H}_d)}} F_{\alpha,z}(\rho, U \sigma U^\ast) = F_{\alpha}^{C}(\lambda^\downarrow(\rho), \lambda^\downarrow(\sigma)).
			\]
			
			For \( z \in (0, \infty) \) and \( \alpha > 1 \), where \( 1 - \alpha \leq 0 \) and \( \lambda_i(\sigma)^{1 - \alpha} \leq \lambda_j(\sigma)^{1 - \alpha} \) for \( 1 \leq i \leq j \leq d \), we apply the Araki-Lieb-Thirring inequality with \( r = q = \alpha \), \( A = \rho \), and \( B = U \sigma^\frac{1-\alpha}{\alpha}U^*\)(\cite{RB,GHHJE}). It follows from Equation \eqref{2-8} that,
			\begin{equation*}
				\begin{aligned}
					F_{\alpha, z}(\rho, U \sigma U^\ast) &\leq \left( \operatorname{Tr} \left( \rho^\alpha U \sigma^{1 - \alpha} U^\ast \right) \right)^{\frac{1}{\alpha}} \\
					&\leq \left( \langle (\lambda^\downarrow(\rho))^\alpha, (\lambda^\uparrow(\sigma))^{1 - \alpha} \rangle \right)^{\frac{1}{\alpha}}.
				\end{aligned}
			\end{equation*}
			Therefore, from Equation \eqref{2-6}, we have
			\begin{equation*}
				\begin{aligned}
					&\max_{\substack{U \in U\mathcal{(H}_d)}}F_{\alpha, z}(\rho, U\sigma U^*) \\&\leq \max_{\substack{U \in U\mathcal{(H}_d)}}\left( \operatorname{Tr} \left( \rho^{\alpha} U \sigma^{1-\alpha} U^* \right) \right)^{\frac{1}{\alpha}} \\
					&\leq \left( \langle (\lambda^{\downarrow}(\rho))^\alpha, (\lambda^{\uparrow}(\sigma))^{1-\alpha} \rangle \right)^{\frac{1}{\alpha}}.
				\end{aligned}
			\end{equation*}
			When we set \( U W_0 |i\rangle = |d - i + 1\rangle \), the above upper bound will be achieved. 
			\par\vspace{10pt} 
			\noindent\textbf{(II)}. For Equation \eqref{3-10}, when \( z \in (0, 1) \) and \( \alpha \in (0, 1) \), we know that the following relationship holds from Equation \eqref{2-8}, when \(\alpha \in (0,1)\).
			\begin{equation*}
				\begin{aligned}\langle (\lambda^{\downarrow}(\rho))^\alpha, (\lambda^{\uparrow}(\sigma))^{1-\alpha} \rangle \leq  \operatorname{Tr} \left( \rho^{\alpha} U \sigma^{1-\alpha} U^* \right).
				\end{aligned}
			\end{equation*}
			Letting $A = U \sigma^{1-\alpha}U^*$, $B = \rho^{\alpha}$ with $r = \frac{1}{z}, q=1$, by Araki-Lieb-Thirring inequality, then
			\begin{equation*}
				\begin{aligned}
					&F_{\alpha, z}(\rho, U\sigma U^*) \\&= \left( \operatorname{Tr} \left( \rho^\frac{\alpha}{2z} U \sigma^\frac{1-\alpha}{z} U^*\rho^\frac{\alpha}{2z} \right)^z \right)^{\frac{1}{\alpha}} \\ 
					&\geq \left( \operatorname{Tr} \left( \rho^{\alpha} U \sigma^{1-\alpha} U^* \right) \right)^{\frac{1}{\alpha}},z \in (0,1).
				\end{aligned}
			\end{equation*}
			Hence,
			\begin{equation*}
				\begin{aligned}
					\min_{\substack{U}}F_{\alpha, z}(\rho, U\sigma U^*) &\geq \min\left( \operatorname{Tr} \left( \rho^{\alpha} U \sigma^{1-\alpha} U^* \right) \right)^{\frac{1}{\alpha}} \\
					&= \left( \langle (\lambda^{\downarrow}(\rho))^\alpha, (\lambda^{\uparrow}(\sigma))^{1-\alpha} \rangle \right)^{\frac{1}{\alpha}}.
				\end{aligned}
			\end{equation*} 
			The above lower bound will be got by taking $UW_0|i\rangle = |d-i+1\rangle $. Hence, we prove that Equation \eqref{3-10} holds for \( z \in (0,1) \) and \( \alpha \in (0,1) \).
			\par\vspace{10pt} 
			Similarly, for $z\in(0,1)$ and $z\in(1,\infty)$, we know that when \(\alpha \in (1,\infty)\).
			\begin{equation*}
				\begin{aligned}
					\langle (\lambda^{\downarrow}(\rho))^\alpha, (\lambda^{\downarrow}(\sigma))^{1-\alpha} \rangle
					\leq \operatorname{Tr} \left( \rho^{\alpha} U \sigma^{1-\alpha} U^*  \right). 
				\end{aligned}
			\end{equation*}
			Since,
			\begin{equation*}
				\begin{aligned}
					F_{\alpha, z}(\rho, U\sigma U^*) &= \left( \operatorname{Tr} \left( \rho^\frac{\alpha}{2z} U \sigma^\frac{1-\alpha}{z} U^*\rho^\frac{\alpha}{2z} \right)^z \right)^{\frac{1}{\alpha}} \\ 
					&\geq \left( \operatorname{Tr} \left( \rho^{\alpha} U \sigma^{1-\alpha} U^* \right) \right)^{\frac{1}{\alpha}}\text{,}
				\end{aligned}
			\end{equation*}
			we obtain
			\begin{equation*}
				\begin{aligned}
					\min_{\substack{U}}F_{\alpha, z}(\rho, U\sigma U^*) &\geq \min \left( \operatorname{Tr} \left( \rho^\alpha U \sigma^\frac{1-\alpha}{z} U^* \right)\right)^{\frac{1}{\alpha}} \\ 
					&= \left( \langle (\lambda^{\downarrow}(\rho))^\alpha, (\lambda^{\downarrow}(\sigma))^{1-\alpha} \rangle \right)^{\frac{1}{\alpha}}.
				\end{aligned}
			\end{equation*}
			The above lower bound will be derived by taking $U = W_0^*$.
			\par \par\vspace{10pt} 
			\noindent\textbf{(III)}. For Equation \eqref{3-11}, when \(\alpha, z \in  \{(\alpha, z) \mid 1 < \alpha \leq 2, \frac{\alpha}{2}\leq z \leq \alpha \}\bigcup \{(\alpha, z) \mid 2 \leq \alpha < \infty, \alpha-1 \leq z \leq \alpha\}\), we first assume that \(\bar{U}_0\) is a unitary obrit such that
			\begin{equation*}
				F_{\alpha,z}(\rho,\bar{U}_0\sigma\bar{U}_0^*)=\min_{\substack{U}}F_{\alpha,z}(\rho,U\sigma U^*)\text{,}
			\end{equation*}
			then,
			\begin{equation*}
				F_{\alpha,z}(\rho,\bar{U}_0\sigma\bar{U}_0^*) \leq F_{\alpha,z}(\rho,W_0^*\sigma W_0)\text{,}
			\end{equation*}
			where \({W_0^*\sigma W_0 = \sum_{i=1}^{d}\lambda_i^{ \downarrow}(\sigma)\mid i\rangle \langle i\mid} \text{ and }W_0^*\sigma W_0= 0\). Hence, 
			\begin{equation}\label{3-12}
				\begin{aligned}
					F_{\alpha,z}(\rho,\bar{U}_0\sigma\bar{U}_0^*) &\leq F_{\alpha,z}(\rho,W_0^*\sigma W_0) 
					\\&= F_{\alpha}^{C}(\lambda^{\downarrow}(\rho),(\lambda^{\downarrow}(\sigma)).
				\end{aligned}
			\end{equation}
			In addition, for \(\alpha\in (1,\infty)\), denote \(\mathcal{S}=\{\tilde{p}=(\tilde{p}_1,\tilde{p}_2, \ldots, \tilde{p}_d) \in \mathbb{R}^d: \tilde{p}_1 \geq \tilde{p}_2 \geq \cdots \geq \tilde{p}_d \}\), one can define a Schur concave (see Proposition \ref{pp2.1}) function \(\varphi : \mathcal{S} \to \mathbb{R}\) by
			\[
			\phi(\tilde{q}) = \sum_{i=1}^{d} \tilde{p}_i^{\alpha} \tilde{q}_i^{1-\alpha}\text{,}
			\]
			thus, for probability distributions \(\tilde{p}\), \(\tilde{q}\), and \(\tilde{q}' \in \mathcal{S}\), we have
			\begin{equation}\label{3-13}
				\begin{aligned}
					F_\alpha^C(\tilde{p},\tilde{q})\geq  F_\alpha^C(\tilde{p},\tilde{q}'), \text { if } \tilde{q} \leq \tilde{q}'.
				\end{aligned}
			\end{equation}
			Moreover, defining a quantum channel on \(\mathcal{D(H}_d)\) as follows:
			\[
			\Phi(K):=\sum_{i=1}^{d}\langle i\mid K\mid i\rangle  \langle i\mid, K\in \mathcal{D(H}_d).
			\]
			Noting that \(\Phi(\rho) = \rho\) and \([\rho, \Phi(K)] = 0\). Thus, the Data processing inequality yields,
			\begin{align*}
				F_{\alpha, z}(\rho, \bar{U}_0 \sigma \bar{U}_0^*) &\geq F_{\alpha, z}(\rho, \Phi(\bar{U}_0 \sigma \bar{U}_0^*)) \\
				&= F_{\alpha}^{C}(\lambda^{\downarrow}(\rho), D_{\bar{u}_0} \lambda^{\downarrow}(\sigma)) \\
				&\geq F_{\alpha}^{C}(\lambda^{\downarrow}(\rho), (D_{\bar{u}_0} \lambda^{\downarrow}(\sigma))^\downarrow)\text{,}
			\end{align*}
			where the first equation from \([\rho,\Phi(\bar{U}_0 \sigma \bar{U}_0^*)] = 0\) and Equation \eqref{2-2}. Meanwhile, the last inequality from Equation \eqref{2-8}, where \(D_{\bar{u}_0}\) is a \(d \times d\) doubly stochastic matrix defined by \((D_{\bar{u}_0})_{ij}=\mid\langle i\mid\bar{U}_0 W_0\mid j\rangle\mid^2\).
			\par It is clear that \(D_{\bar{u}_0} \lambda^{\downarrow}(\sigma)\leq \lambda^{\downarrow}(\sigma)\). Thus, by Equation \eqref{3-13}, we can get 
			\begin{equation}\label{3-14}
				\begin{aligned}
					F_\alpha^C(\lambda^{\downarrow}(\rho),(D_{\bar{u}_0} \lambda^{\downarrow}(\sigma))^\downarrow)\geq  F_\alpha^C(\lambda^{\downarrow}(\rho),\lambda^{\downarrow}(\sigma)).
				\end{aligned}
			\end{equation}
			Combing Equations \eqref{3-12} and \eqref{3-14}, we give rise to 
			\[
			\min_{\substack{U}\in U\mathcal{(H}_d)}F_{\alpha, z}(\rho, U\sigma U^*) = F_\alpha^C(\lambda^{\downarrow}(\rho),\lambda^{\downarrow}(\sigma))\text{,}
			\]
			thus, \( \min_{\substack{U}\in U(H_d)}F_{\alpha, z}(\rho, U\sigma U^*)\) is deduced with \(U= W_0^*\).
		\end{proof}
		In quantum mechanics, a system evolves via unitary transformations generated by a Hamiltonian \( H \), represented as 
		\(U_t = e^{i t H}\) for \( t \in \mathbb{R} \). Unitary operators are continuous in \( t \), and \( U_t \) is path-connected to the identity operator through
		\(U_t = e^{t L}\) where \( L \) is a skew-Hermitian matrix. Given this unitary evolution, we analyze the function 
		\(F_\alpha(\rho, U_t \sigma U_t^\dagger)\), 
		which encapsulates certain properties of the states \( \rho \) and \( \sigma \) transformed under the action of \( U_t \). The continuity of \( U_t \) implies that \( F_\alpha(\rho, U_t \sigma U_t^\dagger) \) inherits this continuity with respect to \( t \).

		\begin{theorem} \label{th3.2}
			For \(\rho,\sigma\in \mathcal{D(H}_d))\), the set \(F_{\alpha, z}(\rho, U\sigma U^*), U\in U\mathcal{(H}_d))\) is identical to the interval, When \(\alpha, z \in \{(\alpha, z) \mid 1 < \alpha \leq 2, \frac{\alpha}{2}\leq z \leq \alpha \}\bigcup \{(\alpha, z) \mid 2 \leq \alpha < \infty, \alpha-1 \leq z \leq \alpha\}\),
			\begin{equation*}
				\begin{aligned}
					[F_\alpha^C(\lambda^{\downarrow}(\rho),\lambda^{\downarrow}(\sigma)), F_\alpha^C(\lambda^{\downarrow}(\rho),\lambda^{\uparrow}(\sigma))]\text{,}
				\end{aligned}
			\end{equation*}
			or
			\begin{equation*}
				\begin{aligned}
					[F_\alpha^C(\lambda^{\downarrow}(\rho),\lambda^{\uparrow}(\sigma)), F_\alpha^C(\lambda^{\downarrow}(\rho),\lambda^{\downarrow}(\sigma))], z,\alpha\in (0,1)\text{,}
				\end{aligned}
			\end{equation*}
		\end{theorem}
		\begin{proof}
			According to Theorem \ref{th3.1}, the function \( F_{\alpha, z}(\rho, U \sigma U^*) \) achieves its maximum at \( U = e^{L_1} \) and its minimum at \( U = e^{L_0} \), where \( L_0 \) and \( L_1 \) are skew-symmetric operators obtained from Stone's theorem(\cite{JBC}).
			\par 
			To investigate further, consider the following function of \( t \),
			\[
			t \mapsto F_{\alpha, z} \left( \rho, e^{(1-t)L_0 + tL_1} \sigma e^{-(1-t)L_0 - tL_1} \right).
			\]
			To apply the intermediate value theorem, it is sufficient to show that this function is continuous with respect to \( t \) on the interval \([0, 1]\). Define the operator \( A_t \) as,
			\[
			A_t = \rho^\frac{\alpha}{2z} U_t \sigma^\frac{1-\alpha}{z} U_t^* \rho^\frac{\alpha}{2z}\text{,}
			\]
			where \( U_t = e^{(1-t)L_0 + tL_1} \) represents a continuous path in the unitary group. Then,
			\[
			\text{Tr}(A_t^z) = F_{\alpha, z} \left( \rho, U_t \sigma U_t^* \right).
			\]
			Thus, to establish continuity of \( t \mapsto \text{Tr}(A_t^z) \), it suffices to demonstrate that this function is continuous with respect to \( t \) for \( z \in (0, \infty) \). Without loss of generality, we assume that all involved operators are positive definite. This will be verified by showing that \( t \mapsto \text{Tr}(A_t^z) \) is continuous, leveraging the fact that \( U_t \) forms a continuous path in the unitary matrix group.
			\par Firstly, by Stone’s theorem, for all \( \forall\varepsilon > 0 \), there exists a \(\delta > 0 \) such that \(\|U_{t+\delta} - U_t\|_1\leq\frac{\varepsilon}{2}\). 
			Applying the submultiplicativity of norm(\cite{JW}), we have
			\begin{align*}
				&\|A_{t+\delta} - A_t\|_1 \\&= 
				\|\rho^{\frac{\alpha}{2z}} \Big[(U_{t+\delta} - U_t)\sigma^{\frac{1-\alpha}{z}}U_{t+\delta}^* + U_t\sigma^{\frac{1-\alpha}{z}}\\&(U_{t+\delta}^* - U_t^*)\Big]\rho^{\frac{\alpha}{2z}}\|_1 \\
				&\leq \|\rho^{\frac{\alpha}{2z}}\|_\infty^2 \Big(\|U_{t+\delta}-U_t\|_1 \cdot \|\sigma^{\frac{1-\alpha}{z}}U_{t+\delta}^*\|_\infty \\&+ \|U_t\sigma^{\frac{1-\alpha}{z}}\|_\infty \cdot \|U_{t+\delta}^*-U_t^*\|_1\Big) \\
				&\leq \|\rho^{\frac{\alpha}{2z}}\|_\infty^2 \|\sigma^{\frac{1-\alpha}{z}}\|_\infty \Big(\|U_{t+\delta} - U_t\|_1 \\&+ \|U_{t+\delta}^* - U_t^*\|_1\Big) \\
				&\leq C \varepsilon \text{,}
			\end{align*}
			where \( C \) is a constant that depends on \(\alpha\) and \( z \).
			\par Furthermore, for \( A \in \mathcal{P}\mathcal{(H}_d)) \) and \( z \in (0, 1) \), we know
			\[
			A^z= \frac{\sin(z\pi)}{\pi}\int_{0}^{\infty}x^z(\frac{1}{x}-\frac{1}{x+A})dx
			\]
			By \( A^{-1} - B^{-1} = A^{-1}(B - A)B^{-1} \), we have	\begin{align*}
				&\|A_{t +\delta}^z -U_t^z\|_1 \\&= \|\frac{\sin(z\pi)}{\pi}\int_{0}^{\infty}x^z(\frac{1}{x+A_t}-\frac{1}{x+A_{t+\delta}})dx\|_1 \\&= \|\frac{\sin(z\pi)}{\pi}\int_{0}^{\infty}x^z\frac{1}{x+A_t}(A_{t+\delta}-A_t)\frac{1}{x+A_{t+\delta}}dx\|_1 \\&\leq \frac{1}{\pi}\int_{0}^{\infty}x^z E dx.
			\end{align*}
			where \(E=\|\frac{1}{x+A_t}\|_
			\infty\cdot\|A_{t+\delta}-A_t\|_1  \cdot \|A_{t+\delta}-A_t\|_\infty\), and for the operator \({A_t>0}\), taking \( \lambda(A_t) = \min\{\lambda_1(A_t), \lambda_2(A_t), \ldots, \lambda_d(A_t)\}\), then
			\[
			\|\frac{1}{x+A_t}\|_\infty = \frac{1}{x+\lambda_d^\downarrow(A_t)} \text{, }x\geq 0.
			\]
			Hence
			\begin{align*}
				&\|(A_{t+\delta}^z - A_t^z)\|_1\\ &\leq \frac{1}{\pi} \|A_{t+\delta} - A_t\|_1 \int_{0}^{\infty} x^z \left(\frac{1}{x + \alpha}\right)^2 dx \\
				&\leq \frac{1}{\pi} \varepsilon \int_{0}^{\infty} x^z \left(\frac{1}{x + \alpha}\right)^2 dx \\
				&\leq \frac{1}{\pi} \varepsilon \int_{0}^{\infty} \frac{1}{y^{2-z}}dy :=b\varepsilon\text{,}
			\end{align*}
			\noindent where denote \( y= x + \alpha \) , \(\alpha = \min\{\lambda_d^\downarrow(A_t), \lambda_d^\downarrow(A_{t+\delta})\}\text{, and } b= \frac{1}{\pi} \int_a^{\infty} \frac{1}{y^{2-z}}dy \). It is easy to see that the integral \( \int_a^{\infty} \frac{1}{y^{2-z}}dy < \infty \) with \( z \in (0, 1) \). Thus
			\[
			|\operatorname{Tr} (A_{t+\delta}^z) - \operatorname{Tr} (A_t^z)| \leq \|A_{t+\delta}^z - A_t^z\|_1\leq b\varepsilon\text{,}
			\]
			\noindent hence the function \( t\to \operatorname{Tr}(A_t^z)\) is continuous with respect to \( t \) when \( z \in (0, 1) \).
			\par Indeed, \( \operatorname{Tr}(A_t^z)\) is continuous with respect to \( t \) for \( z > 0 \). The following inequality shows that \( \operatorname{Tr}(A_t^{2z})\) is continuous in \( t \) when \( z \in (0, 1) \),
			\begin{align*}
				&|\operatorname{Tr} (A_{t+\delta}^{2z}) - (A_t^{2z})| \\& \leq \|A_{t+\delta}^{2z} - A_t^{2z}\|_1 \\& = \|A_{t+\delta}^zA_{t+\delta}^z - A_t^zA_t^z\|_1
				\\&\leq \|A_{t+\delta}^z(A_{t+\delta}^z - A_t^z)\|_1 +  \|(A_{t+\delta}^z - A_t^z)\|_1 \\&\leq (\|A_{t+\delta}^z\|_\infty + \|A_t^z\|_\infty)\|(A_{t+\delta}^z - A_t^z)A_t^z)\|_1
				\\&\leq c\varepsilon\text{,}
			\end{align*}
			\noindent where \( c = b (\|A_{t+\delta}^z\|_\infty + \|A_t^z\|_\infty) < \infty \). Thus, \( \operatorname{Tr}(A_t^{2z}) \) is continuous with respect to \( t \) for \( z \in (0, 1) \). That is, \( \operatorname{Tr}(A_t^z) \) is continuous in \( t \) for \( z \in (0, 2) \). By repeating the previous method, we can show that \( \operatorname{Tr}(A_t^z) \) is continuous with respect to \( t \) when \( z \in (0, 2^n) \), \( n \in \mathbb{N} \). Hence, \( \operatorname{Tr}(A_t^z) \) is continuous with respect to \( t \) for each \( z > 0 \), which shows our conclusion.
		\end{proof}
		
		\begin{remark}
			
			The continuity of \( t \to \operatorname{Tr}(A_t^z) \) with respect to \( t \) for \(z > 0\) is crucial for applying the intermediate value theorem. The detailed proof leverages properties of matrix norms and functional calculus for positive definite operator. Notably, the argument demonstrates that the continuity holds across all positive \( z \), showing the robustness of the result. The method relies on bounding differences in matrix functions and norms, which ensures that small changes in \( t \) lead to small changes in the trace function, confirming the desired continuity.

			For \( \alpha \in (0, 1) \cup (1, \infty) \) and \( z \in (0, \infty) \), the sandwiched quantum \( \alpha \)-$z$-R\'enyi Relative Entropy is defined in Ref\cite{KMRA}.
			\begin{equation*}
				\begin{aligned}
					&S_{\alpha,z}(\rho\|\sigma)\\&=
					\left\{
					\begin{array}{ll}
						\frac{\alpha}{\alpha-1}\log F_{\alpha,z}(\rho,\sigma)\text{,}& \text{supp}(\rho)\subseteq \text{supp}(\sigma)\text{,} \\
						\infty\text{,} & \text{otherwise} .
					\end{array}
					\right.
				\end{aligned}
			\end{equation*}
			For \( \tilde{p}=\) \(\{\tilde{p}_1, \ldots, \tilde{p}_d\}  \text{ and } \tilde{ q}= \{\tilde{q}_1, \ldots, \tilde{q}_d\} \), two probability distributions, the classical \(\alpha\)-Rényi relative entropy is described in Ref\cite{AR}.
			\begin{equation*}
				\begin{aligned}
					S_{\alpha,z}^C(\tilde{p}\|\tilde{q}) = \frac{\alpha}{\alpha-1}\log F_\alpha^C(\tilde{p},\tilde{q}) = 	S_\alpha^C(\tilde{p}\|\tilde{q}).
				\end{aligned}
			\end{equation*}
			\noindent Easily, we have the following consequence.
		\end{remark}
		
		\begin{corollary}\label{cl3.3} For two states $\rho,\sigma\in \mathcal{D(H}_d)$, where $\sigma$ is full-ranked, we have
			\begin{equation*}
				\begin{aligned}
					\max_{U\in U\mathcal{(H}_d)}S_{\alpha,z}(\rho\|U_\sigma) =S_{\alpha,z}^C(\lambda^\downarrow(\rho)\|\lambda^\uparrow(\sigma))\text{,}
				\end{aligned}
			\end{equation*}
			for \(\alpha,z\in \{(\alpha,z)|\alpha,z\in (0,1)\}\cup\{\alpha,z\in (1,\infty)\text{,} z\in(0,\infty)\}\).
			\begin{equation*}
				\begin{aligned}
					\min_{U\in U\mathcal{(H}_d)}S_{\alpha,z}(\rho\|U_\sigma) =S_{\alpha,z}^C(\lambda^\downarrow(\rho)\|\lambda^\downarrow(\sigma)), 
				\end{aligned}
			\end{equation*}
			\(\text{and for, }\alpha, z \in \{ (\alpha, z) \mid \alpha \in (0,1),(1,\infty), z \in (0, 1)\} \cup \{(\alpha, z) \mid 1 < \alpha \leq 2, \frac{\alpha}{2}\leq z \leq \alpha \}\bigcup \{(\alpha, z) \mid 2 \leq \alpha < \infty, \alpha-1 \leq z \leq \alpha\}\).
		\end{corollary}
		\begin{remark} 
			The above implications and theorem elucidate the relationship between the sandwiched quantum $\alpha$-$z$-Rényi relative entropy $S_{\alpha,z}(\rho\|\sigma)$ and its classical counterpart $S_{\alpha,z}^C(\tilde{p}\|\tilde{q})$. 
			For different ranges of $\alpha$ and $z$, the maximum and minimum values of $S_{\alpha,z}(\rho\|U_\sigma)$ are expressed in terms of the classical $\alpha$-Rényi relative entropy of the eigenvalue distributions of $\rho$ and $\sigma$. This highlights the close interplay between quantum and classical Rényi relative entropy under various parameter conditions.
		\end{remark}
		\section{$\alpha$-z-fidelity of other quantum channel evolutions}
		Inspired by the unitary orbits mentioned above, we will consider the extremums under the evolution of several equally important orbits. For example, the focus is mainly on the following three physically realizable channels:
		\par
		(1)  All quantum channels. Completely positive trace-preserving linear maps.
		\par
		(2)  Unital quantum channels. A quantum channel \(\Phi\) is called unital if \(\Phi(I) = I\).
		\par
		(3) Mixed unitary channels. Where \(U_i\) are unitary operators and \(p_i\) are probabilities, a quantum channel \(\Psi\) is called mixed unitary if \(\Psi(\cdot) = \sum_i p_i U_i (\cdot) U_i^\dagger\).
		\par
		The main work of this section is to study the extremal value problem of arbitrary quantum states \(\rho\) and \(\sigma\) under different quantum channel evolutions.
		\par 
		Additionally, using the closed convex sets of density operators and quantum channels, a new geometric interpretation of the \(\alpha\)-$z$-fidelity is provided. By treating density operators as normalized projections on subspaces, we further study the geometric property of 
		\(\alpha\)-$z$-fidelity as a measure of distance between two spaces.
		\par 
		Considering the three types of quantum channels mentioned above, we can determine the extremal values of their \(\alpha\)-$z$-fidelity. We focus on unital quantum channels. Using optimization theory, we find the extremal values of the quantum \(\alpha\)-$z$-fidelity \(F_{\alpha,z}(\rho, \Phi(\sigma))\) between  quantum states \(\rho\) and \(\sigma\), where \(\sigma\) is transmitted through any unital quantum channel.
		
		
		\begin{theorem}\label{th4.1}
			Assuming that \(\Phi\) is a unital quantum channel, the maximum and minimum values of the \(\alpha\)-z-fidelity between \(\rho\) and \(\Phi(\sigma)\) are as follows
			when \(\alpha \in (0, 1), z\in(0, \infty)\),
			\begin{equation}\label{4-1}
				\begin{aligned}
					F^\text{C}_{\alpha,z}(\lambda_j^\uparrow(\rho), \lambda_j^\downarrow(\Phi(\sigma))) \geq F^\text{C}_{\alpha,z}(\lambda_j^\uparrow(\rho), \lambda_j^\downarrow(\sigma)), 
				\end{aligned}
			\end{equation}
			when \( \alpha \in (1, \infty), z\in(0, \infty)\),
			\begin{equation}
				\begin{aligned}
					F^\text{C}_{\alpha,z}(\lambda_j^\uparrow(\rho), \lambda_j^\downarrow(\Phi(\sigma))) \leq F^\text{C}_{\alpha,z}(\lambda_j^\uparrow(\rho), \lambda_j^\downarrow(\sigma)),
				\end{aligned}
			\end{equation}
			
			The extremal values are taken over all unital quantum channels \(\Phi\).
		\end{theorem}
		
		\begin{proof}
			It follows from Proposition \ref{pp2.1} that \( F_{\alpha,z} \) shows concavity for \( \alpha \in (0,1) \) and \( z \in (0,\infty) \), while it demonstrates  convexity for \( \alpha \in (1,\infty) \) and \( z \in (0,\infty) \). Furthermore, when the quantum states \( \rho \) and \( \sigma \) can be simultaneously diagonalized, the quantum \( \alpha \)-\( z \)-fidelity simplifies to the classical \( \alpha \)-fidelity. Thus, 
			\begin{equation}\label{4-2}
				\begin{aligned}
					F_{\alpha,z}^{C}(\rho,\sigma) = F_{\alpha}^{C}(\rho,\sigma)
				\end{aligned}
			\end{equation}
			for any quantum states \( \rho \) and \( \sigma \).
			
			Considering the above Proposition \ref{pp2.1} (6), let \(\phi^\alpha(w) = \sum_i p_i^\alpha w_i^{1-\alpha}\), \(w \in D\) and \(p_1 \leq \cdots \leq p_d\), where \(\phi\) is a real-valued function defined and continuous on D. Clearly, \(\phi^\alpha_{(i)}(w) = (1 - \alpha), (\frac{p_i}{w_i})^\alpha\), when \(\alpha \in (0, 1)\)
			\begin{equation*}
				0 \geq -\phi^\alpha_{(1)}(w) \geq -\phi^\alpha_{(2)}(w) \geq \cdots \geq -\phi^\alpha_{(d)}(w),
			\end{equation*}
			when \( \alpha \in (1, \infty)\),
			\begin{equation*}
				0 \geq \phi^\alpha_{(1)}(w) \geq \phi^\alpha_{(2)}(w) \geq \cdots \geq \phi^\alpha_{(d)}(w),
			\end{equation*}
			
			If \(\tilde{p} \prec^w \tilde{q}\), then
			\begin{equation}\label{4-3}
				\begin{aligned}
					& \phi^\alpha(\tilde{p}) \geq \phi^\alpha\tilde{q}, \quad \alpha \in (0, 1) 
					\\
					&\phi^\alpha\tilde{p}) \leq \phi^\alpha(\tilde{q}), \quad \alpha \in (1, \infty)
				\end{aligned}
			\end{equation}
			Define a function \(f\) as follows,
			\begin{equation}\label{4-4}
				\begin{aligned}
					f_\alpha(q_1, \ldots, q_d) := \sum_j (p_j^\uparrow)^\alpha (q_j^\downarrow)^{1-\alpha}, \quad \alpha \in (0, \infty)
				\end{aligned}
			\end{equation}
			where \(\{p_j\}_{j=1}^d\)represents a fixed probability distribution with elements summing to one.
			\par Using Equation \eqref{4-3}, it is evident that for \(\alpha \in (0, 1)\), the function \(f\) is concave, while for \(\alpha \in (1, \infty)\), the function \(f\) is convex. 
			\par 
			Additionally, for quantum states \(\rho\) and \(\sigma\), and a unital quantum channel \(\Phi\), Uhlmann’s Theorem implies that \(\Phi(\sigma) \prec \sigma\). In Equation  \eqref{4-4}, if we set \(p_j = \lambda_j(\rho)\) and \(q_j = \lambda_j(\sigma)\), then when \(\alpha \in (0, 1)\),
			\begin{equation*}
				\begin{aligned}
					&\sum_j (\lambda_j^\uparrow(\rho))^\alpha (\lambda_j^\downarrow(\sigma))^{1-\alpha} \\&\leq \sum_j (\lambda_j^\uparrow(\rho))^\alpha (\lambda_j^\downarrow(\Phi(\sigma)))^{1-\alpha}, 
				\end{aligned}
			\end{equation*}
			
			when \(\quad \alpha \in (1, \infty)\),
			\begin{equation*}
				\begin{aligned}
					&\sum_j (\lambda_j^\uparrow(\rho))^\alpha (\lambda_j^\downarrow(\sigma))^{1-\alpha} \\&\geq \sum_j (\lambda_j^\uparrow(\rho))^\alpha (\lambda_j^\downarrow(\Phi(\sigma)))^{1-\alpha}. 
				\end{aligned}
			\end{equation*}
			Hence, when \(\alpha \in (0, 1)\)
			\[
			F^\text{C}_{\alpha}(\lambda_j^\uparrow(\rho), \lambda_j^\downarrow(\Phi(\sigma))) \geq F^\text{C}_{\alpha}(\lambda_j^\uparrow(\rho), \lambda_j^\downarrow(\sigma)), 
			\]
			when \(\alpha \in (1, \infty)\)
			\[
			F^\text{C}_{\alpha}(\lambda_j^\uparrow(\rho), \lambda_j^\downarrow(\Phi(\sigma))) \leq F^\text{C}_{\alpha}(\lambda_j^\uparrow(\rho), \lambda_j^\downarrow(\sigma)),
			\]
			For \(z > 0 \text{, }
			F_{\alpha,z}^{C}(\rho, \sigma) = F_{\alpha}^{C}(\rho, \sigma)  
			\), Equation \eqref{4-1} is thereby proven.
		\end{proof}
		\begin{remark}
			Although this result is related to Theorem \ref{th3.1}, since there exist unital quantum channels that are not convex combinations of unitary transformations, Theorem \ref{th4.1} is not a direct consequence of Theorem  \ref{th3.1}. 
			\par Inspired by \cite{LJP}, where Jin Li et al. studied the minimum fidelity between a fixed state \(\rho\) and a state \(\sigma\) transmitted through quantum channels, we extend their work by considering a broader scope with more factors. In their work, they analyzed various types of quantum channels, provided a geometric interpretation related to density matrices and quantum channels, and studied fidelity as a measure of distance between subspaces, along with its connection to regular angles between these subspaces. In contrast, we focus on the \(F_{a,z}\) quantum fidelity between two quantum states \(\rho\) and \(\sigma\), where we investigate extremal problems under different quantum channel evolutions across a wider range of \((a,z)\). Additionally, we offer a new geometric interpretation of \(F_{a,z}\) within this extended range.

		\end{remark}
		
		\begin{corollary}\label{cl4.2}
			Suppose that \( \mathcal{S(H}_d)) \subseteq \mathcal{D(H}_d))\) is the set of all pure states. Let \(\rho \in \mathcal{D(H}_d))\). When \(\alpha,z \in \{(\alpha,z)\mid 0 < \alpha < 1, z \geq \max\{\alpha, 1 - \alpha\}\} \),
			\begin{equation}\label{4-5}
				\begin{aligned}
					\min_{\sigma \in \mathcal{S(H}_d)} F_{\alpha,z}(\rho, \sigma) = \lambda_{\min}(\rho),
				\end{aligned}
			\end{equation}
			when \( \alpha,z\in \{(\alpha,z) \mid 1 < \alpha \leq 2, \frac{\alpha}{2} \leq z \leq \alpha\}\bigcup \{(\alpha,z) \mid2 \leq \alpha < \infty, \alpha - 1 \leq z \leq \alpha\}\)
			\begin{equation}\label{4-6}
				\begin{aligned}
					\max_{\sigma \in \mathcal{S(H}_d)} F_{\alpha,z}(\rho, \sigma) = \lambda_{\max}(\rho).
				\end{aligned}
			\end{equation}
			where \(\lambda_{\min}(\rho)\) and \(\lambda_{\max}(\rho)\) denote the minimum and maximum eigenvalues of \(\rho\), respectively.
		\end{corollary}
		
		\begin{proof}
			It is well known that any mixed state \(\sigma\) can be represented as a convex combination of pure states:
			\[
			\sigma = \sum_i q_i |\psi_i\rangle \langle \psi_i|,
			\]
			where \( q_i \geq 0 \) and \(\sum_i q_i = 1\). 
			
			\par From Proposition \ref{pp2.1}, we can deduce that when \(\alpha, z \in \{(\alpha, z) \mid 0 < \alpha < 1, z \geq \max\{\alpha,1-\alpha\}\}\),
			\begin{equation*}
				F_{\alpha,z}(\rho, \sigma) \geq \sum_i q_i F_{\alpha,z}(\rho, |\psi_i\rangle \langle \psi_i|),
			\end{equation*}
			When \( \{(\alpha, z) \mid 1 < \alpha \leq 2, \frac{\alpha}{2}\leq z \leq \alpha \}\bigcup \{(\alpha, z) \mid 2 \leq \alpha < \infty, \alpha-1 \leq z \leq \alpha,\}\)
			\begin{equation*}
				F_{\alpha,z}(\rho, \sigma) \leq \sum_i q_i F_{\alpha,z}(\rho, |\psi_i\rangle \langle \psi_i|),
			\end{equation*}
			\noindent Hence, for \(\alpha, z \in \{(\alpha, z) \mid 0 < \alpha < 1, z \geq \max\{\alpha,1-\alpha\}\}\), there exists at least one pure state \(|\psi_0\rangle\) such that,
			\[
			\langle \psi_0 | \rho | \psi_0 \rangle = F_{\alpha,z}(\rho, |\psi_0\rangle \langle \psi_0|) \leq F_{\alpha,z}(\rho, \sigma).
			\]
			Similarly, \( \{(\alpha, z) \mid 1 < \alpha \leq 2, \frac{\alpha}{2}\leq z \leq \alpha \}\bigcup \{(\alpha, z) \mid 2 \leq \alpha < \infty, \alpha-1 \leq z \leq \alpha,\}\), there exists at least one pure state \(|\psi_1\rangle\) such that,
			\[
			\langle \psi_1 | \rho | \psi_1 \rangle = F_{\alpha,z}(\rho, |\psi_1\rangle \langle \psi_1|) \geq F_{\alpha,z}(\rho, \sigma).
			\]
			By the Courant-Fischer Theorem(\cite{IYM}), we can choose an eigenvector \(|\tilde{\psi}_0\rangle\) such that \(\rho\) attains its minimum eigenvalue and an eigenvector \(|\tilde{\psi}_1\rangle\) such that \(\rho\) attains its maximum eigenvalue. Thus,
			\[
			|\tilde{\psi}_0\rangle \langle \tilde{\psi}_0| \in\mathcal{S(H}_d)), \quad |\tilde{\psi}_1\rangle \langle \tilde{\psi}_1| \in \mathcal{S(H}_d)),
			\]
			and
			\[
			F_{\alpha,z}(\rho, \sigma) \geq \langle \tilde{\psi}_0 | \rho | \tilde{\psi}_0 \rangle = \lambda_{\min}(\rho),
			\]
			\[
			F_{\alpha,z}(\rho, \sigma) \leq \langle \tilde{\psi}_1 | \rho | \tilde{\psi}_1 \rangle = \lambda_{\max}(\rho),
			\]
			Therefore, Equations \eqref{4-5} and \eqref{4-6} can be established.
		\end{proof}
		
		\begin{remark}
			The values of \( \lambda_{\min}(\rho) \) and \( \lambda_{\max}(\rho) \) indicate how close or far \( \rho \) is from the maximally mixed state. In other words, they reflect how close or far \( \rho \) is from the endpoint states (Singular density operators).
		\end{remark}
		
		\begin{corollary}\label{cl4.3}
			Let \(\rho\) and \(\sigma\) be two quantum states, and \(\Phi\) be a quantum channel. That is, when \(\alpha, z \in \{(\alpha, z) \mid 0 < \alpha < 1, z \geq \max\{\alpha,1-\alpha\}\}\),
			\begin{equation}\label{4-7}
				\begin{aligned}
					\min_{\Phi} F_{\alpha,z}(\rho, \Phi(\sigma)) = \lambda_{\min}(\rho).
				\end{aligned}
			\end{equation}
			When \( \{(\alpha, z) \mid 1 < \alpha \leq 2, \frac{\alpha}{2}\leq z \leq \alpha \}\bigcup \{(\alpha, z) \mid 2 \leq \alpha < \infty, \alpha-1 \leq z \leq \alpha\}\),
			\begin{equation}\label{4-8}
				\begin{aligned}
					\max_{\Phi} F_{\alpha,z}(\rho, \Phi(\sigma)) = \lambda_{\max}(\rho),
				\end{aligned}
			\end{equation}
			where the extrema are taken over all quantum channels \(\Phi\).
		\end{corollary}
		
		\begin{proof} 
			For a quantum channel \(\Phi\), suppose
			\[
			\mathcal{S(\tilde{H}}_d))= \{\Phi(\sigma)\}.
			\]
			If \(
			\Phi(\cdot) = \text{Tr}(\cdot)|\psi\rangle\langle\psi|
			\), the map \(\Phi\) is a completely positive and trace-preserving linear map, so \(\Phi\) is a quantum channel. Additionally,
			\[
			|\psi\rangle\langle\psi| \in \mathcal{S(H}_d)),
			\]
			we have
			\[
			\mathcal{S(H}_d)), \subset \mathcal{S(\tilde{H}}_d)),
			\]
			According to Corollary \ref{cl4.2}, it is clear that the values of Equation \eqref{4-7} and \eqref{4-8} 
			can be determined.
		\end{proof}
		
		\begin{corollary}\label{cl4.4}
			If \(\Phi\) is a quantum channel and \(\Phi(\cdot) = \text{Tr}(\cdot) \rho\), then
			\[
			F_{\alpha,z}(\rho, \Phi(\sigma)) = F_{\alpha,z}(\rho, \sigma) = 1.
			\]
			Therefore, the extrema of \(F_{\alpha,z}(\rho, \Phi(\sigma))\) occur in two cases: when \(\alpha, z \in \{(\alpha, z) \mid 0 < \alpha < 1, z \geq \max\{\alpha,1-\alpha\}\}\), the maximum value of \( F_{\alpha,z}(\rho, \Phi(\sigma)) \) is 1, and when \( \{(\alpha, z) \mid 1 < \alpha \leq 2, \frac{\alpha}{2}\leq z \leq \alpha \}\bigcup \{(\alpha, z) \mid 2 \leq \alpha < \infty, \alpha-1 \leq z \leq \alpha\}\), the minimum value of \( F_{\alpha,z}(\rho, \Phi(\sigma)) \) is 1.
			
		\end{corollary}
		\par  In summary, \(\Phi\) is any quantum channel, finding the extrema of  \(F_{\alpha,z}(\rho, \Phi(\sigma))\) is relatively straightforward. However, once we restrict ourselves to unital or mixed unitary channels, the complexity of the problem increases, as the freedom to choose becomes limited in these special cases. It is also important to note that we lack theoretical results to guide us, making the search for the maximum value a challenging task.
		
		For a unitary-preserving quantum channel \(\Phi\), if
		\[
		\Psi = p\Phi + (1 - p)\Omega,
		\]
		\(\Psi\) is a mixed unitary channel. In fact, the set of all mixed unitary channels is a subset of the set of all unitary-preserving channels, with \(\Omega\) being the central point of both sets.

		\begin{theorem}\label{th4.5}
			Let \(\rho\) and \(\sigma\) be two quantum states, and \(\Psi\) be a mixed unitary channel. when \(\alpha, z \in \{(\alpha, z) \mid 0 < \alpha < 1, z \geq \max\{\alpha,1-\alpha\}\}\), we have, 
			\begin{equation}\label{4-9}
				\begin{aligned}
					\min_{\Psi} F_{\alpha,z}(\rho, \Psi(\sigma)) = F^C_{\alpha,z}(\lambda(\rho)^\downarrow, \lambda(\sigma)^\uparrow),
				\end{aligned}
			\end{equation}
			when \( \{(\alpha, z) \mid 1 < \alpha \leq 2, \frac{\alpha}{2}\leq z \leq \alpha \}\bigcup \{(\alpha, z) \mid 2 \leq \alpha < \infty, \alpha-1 \leq z \leq \alpha\}\),
			\begin{equation}\label{4-10}
				\begin{aligned}
					\max_{\Psi} F_{\alpha,z}(\rho, \Psi(\sigma)) = F^C_{\alpha,z}(\lambda(\rho)\downarrow, \lambda(\sigma)\uparrow),
				\end{aligned}
			\end{equation}
			where the extremas are taken over all mixed unitary channels \(\Psi\).
		\end{theorem}
		\begin{proof} By Proposition \ref{pp2.1}, we know that when \(\alpha, z \in \{(\alpha, z) \mid 0 < \alpha < 1, z \geq \max\{\alpha,1-\alpha\}\}\), the quantum \(\alpha\)-$z$-fidelity is concave, and when \( \{(\alpha, z) \mid 1 < \alpha \leq 2, \frac{\alpha}{2}\leq z \leq \alpha \}\bigcup \{(\alpha, z) \mid 2 \leq \alpha < \infty, \alpha-1 \leq z \leq \alpha\}\), the quantum \(\alpha\)-$z$-fidelity is convex. Therefore, the minimum of \(F(\rho, \Psi(\sigma))\) and the maximum of \(F(\rho, \Psi(\sigma))\) are attained at the endpoints of the set of quantum channels \(\Psi\). Thus, when considering the set of mixed unitary channels, the extrema must occur at a unitary channel (of the form: \(\Psi(\cdot) = U(\cdot)U^*\)). 
			
			Using Theorem \ref{th3.1}, the conclusions from Equations \eqref{4-9} and \eqref{4-10} lead to this proof, where the extremas are taken over all mixed unitary channels \(\Psi\).
		\end{proof}
		
		\begin{remark}
			It is observed that when $\alpha$ and $z$ belong to other intervals, the extremums of \(\alpha\)-z-fidelity may occur at any point, making it difficult to determine the remaining extremas. As in Corollary \ref{cl3.3}, it is easy to provide the expressions for the $\alpha$-$z$-Rényi relative entropy $S_{\alpha,z}(\rho\|\sigma)$ under the evolution of these three quantum channels, respectively.
			
		\end{remark}
		\section{Geometric properties of quantum \(\alpha\)-$z$-fidelitye}
		Geometric interpretation helps intuitively understand the properties of quantum $\alpha$-$z$-fidelity, revealing the distance and similarity between quantum states. It simplifies complex mathematical formulas, making the concept of fidelity easier to grasp, especially through providing a clear framework for quantum state optimization and transformation. This perspective also helps study the topological features of quantum state spaces, offering new theoretical tools for quantum computing and communication.
		Let \( S_m \) and \( S_n \) be \( m \)-dimensional and \( n \)-dimensional subspaces of a \( d \)-dimensional Hilbert space \( H_d \),  respectively. Let \( P_{S_m} \) and \( P_{S_n} \) be the orthogonal projections onto the subspaces \( S_m \) and \( S_n \), respectively. Thus, the density operators are given by
		\(
		\rho_{S_m} = \frac{1}{m} P_{S_m}
		\) and \(
		\rho_{S_n} = \frac{1}{n} P_{S_n}.
		\)
		This section primarily focuses on studying the quantum \(\alpha\)-$z$-fidelity \( F_{\alpha,z}(\rho_{S_m}, \rho_{S_n}) \) and the geometric properties of the two subspaces \( S_m \) and \( S_n \).

		\begin{theorem} \label{th4.6}
			Let \( S_m \) and \( S_n \) denote \( m \)-dimensional and \( n \)-dimensional subspaces of a \( d \)-dimensional Hilbert space \( \mathcal{H}_d \). Define
			\[
			\rho_{S_m} = \frac{1}{m} P_{S_m} \quad \text{and} \quad \rho_{S_n} = \frac{1}{n} P_{S_n}.
			\]
			Then,
			\begin{equation}
				\begin{aligned}
					\frac{\max \left\{m+n-d,0\right\}}{m^z n^{1-\alpha}}
					& \leq F^{z}_{\alpha,z}(\rho_{S_m}, \rho_{S_n}) 
					\\ &\leq \min \left\{  \frac{m^{1-z}}{n^{1-\alpha}},  \frac{n^\alpha}{m^z}  \right\}
				\end{aligned}
			\end{equation}
		\end{theorem}
		\begin{proof}
			First, consider the case where \( \rho_{S_m} \) and \( \rho_{S_n} \) are commuting. In this case, for \( \alpha \in (0, \infty) \) and \(z\in (0,\infty)\),
			\begin{equation}\label{5-2}
				\begin{aligned}
					&F^{z}_{\alpha,z}(\rho_{S_m}, \rho_{S_n}) \\&= \text{Tr} \left[ \left( \frac{1}{m} P_{S_m} \right)^{\frac{1}{2}} \left( \frac{1}{n} P_{S_n} \right)^{\frac{1-\alpha}{z}} \left( \frac{1}{m} P_{S_m} \right)^{\frac{1}{2}} \right]^z
					\\	& = \frac{1}{m^z n^{1-\alpha}} \text{Tr} \left[ \left( P_{S_m} \right)^{\frac{1}{2}} \left( P_{S_n} \right)^{\frac{1-\alpha}{z}} \left( P_{S_m} \right)^{\frac{1}{2}} \right]^z \\
					& =  \frac{1}{m^z n^{1-\alpha}}  \text{Tr} \left[ P_{S_m} P_{S_n} P_{S_m} \right]^z
					\\	& =  \frac{1}{m^z n^{1-\alpha}}  \text{Tr} \left[ P_{S_m \cap S_n} \right]^z
					\\	& = \frac{1}{m^z n^{1-\alpha}}  \dim(S_m \cap S_n).
				\end{aligned}
			\end{equation}
			Using this result, we can determine the extremas of \( F_{\alpha,z}(\rho_{S_m}, \rho_{S_n}) \) by varying \( S_m \) and \( S_n \) over all \( m \)-dimensional and \( n \)-dimensional subspaces of the \( d \)-dimensional Hilbert space \( \mathcal{H}_d \). According to Theorem \ref{th3.1}, the extremas of \( F_{\alpha,z}(\rho_{S_m}, \rho_{S_n}) \) are obtained when \( S_m \) and \( S_n \) are chosen such that \( \rho_{S_m} \) and \( \rho_{S_n} \) commute.

			Given that,
			\[
			\max(m + n - d, 0) \leq \dim(S_m \cap S_n) \leq \min(m, n)
			\]
			Therefore,
			\begin{equation}
				\begin{aligned}
					\frac{\max \left\{m+n-d,0\right\}}{m^z n^{1-\alpha}}
					& \leq F^{z}_{\alpha,z}(\rho_{S_m}, \rho_{S_n}) 
					\\ &\leq \min \left\{  \frac{m^{1-z}}{n^{1-\alpha}},  \frac{n^\alpha}{m^z}  \right\}
				\end{aligned}
			\end{equation}
		\end{proof}
		\begin{remark}
			In Theorem \ref{th4.6}, the quantum $\alpha$-$z$-fidelity ($F_{\alpha,z}(\rho_{S_m}, \rho_{S_n})$) measures the overlap between two subspaces. When $S_m$ and $S_n$ are disjoint, $F_{\alpha,z}(\rho_{S_m}, \rho_{S_n}) = 0$. When $S_m = S_n$, $F_{\alpha,z}(\rho_{S_m}, \rho_{S_n}) = d^{\alpha-z}$.
			
		\end{remark}
		
		\begin{proposition} \label{pp4.7}
			Suppose $S_n$ is an $n$-dimensional subspace of the Hilbert space $\mathcal{H}_d$. For a quantum state $\rho$, let its eigenvalues be arranged in descending order $($i.e., $\lambda_1(\rho) \geq \lambda_2(\rho) \geq \cdots \geq \lambda_{d - 1}(\rho) \geq \lambda_d(\rho))$. Then,
			\begin{equation}\label{5-3}
				\begin{aligned}
					\sum_{i = d - n + 1}^{d} n^{\alpha - 1} \lambda_i^z(\rho) \leq F^{z}_{\alpha,z}(\rho, \rho_{S_n}) \leq \sum_{i = 1}^{n} n^{\alpha-1} \lambda_i^z(\rho)
				\end{aligned}
			\end{equation}
			
			For any fixed $\rho$, there exist subspaces $S_n$ that achieve the upper and lower bounds of $F^{z}_{\alpha,z}(\rho, \rho_{S_n})$.
		\end{proposition}
		
		\begin{proof}
			First, consider the operator $P_{S_n} \rho P_{S_n}$ on the subspace $S_n$, with its eigenvalues denoted as $\tilde{\lambda}_1(\rho), \tilde{\lambda}_2(\rho), \ldots, \tilde{\lambda}_{n - 1}(\rho), \tilde{\lambda}_n(\rho)$ satisfying,
			\[
			\tilde{\lambda}_1(\rho) \geq \tilde{\lambda}_2(\rho) \geq \cdots \geq \tilde{\lambda}_{n - 1}(\rho) \geq \tilde{\lambda}_n(\rho)
			\]
			By the Cauchy interlacing Theorem(\cite{GMST}), we have:
			\[
			\lambda_{i + d - n} \leq \tilde{\lambda}_i \leq \lambda_i
			\]
			From the definition of $\alpha$-z-fidelity, we have:
			\begin{equation}\label{5-4}
				\begin{aligned}
					F^{z}_{\alpha,z}(\rho, \rho_{S_n}) & = \text{Tr} \left[ \left(\frac{1}{n} P_{S_n}\right)^{\frac{1-\alpha}{2z}} \rho \left(\frac{1}{n} P_{S_n}\right)^{\frac{1-\alpha}{2z}} \right]^z 
					\\  & = n^{\alpha - 1} \text{Tr} \left( P_{S_n} \rho P_{S_n} \right)^z.
				\end{aligned}
			\end{equation}
			Therefore Equation \eqref{5-3} is proven.
		\end{proof}
		
		\begin{remark}
			Corollary \ref{cl4.2} is a corollary of Proposition \ref{pp4.7}. In fact, all pure states are rank-1 projections, so the dimension of the set of all pure states is 1. By concavity, the extremas of $\alpha$-$z$-fidelity will be achieved at the boundary of the set $D(\mathcal{H}_d)$, which is at pure states. Hence, we obtain the lower bound $\lambda_d(\rho)$ and upper bound $\lambda_1(\rho)$ in Proposition \ref{pp4.7}, which is exactly the result of Corollary \ref{cl4.2}. 
			
		\end{remark}

		\section{Conclusion}
		In this paper, we define a quantum \(\alpha\)-$z$-fidelity, which generalizes the concept of fidelity. Specifically, we introduce this new concept to study the similarity and distance between quantum states in a broader context. Furthermore, we investiage some of its key properties, including the data processing inequality, which is a crucial tool in quantum information theory.
		\par
		We further study  the \(\alpha\)-\(z\)-fidelity between unitary orbits of any quantum states using mathematical tools such as the data processing inequality, Golden-Thompson inequality, and Araki-Lieb-Thirring inequality. Furthermore, we obtain the explicit expressions for the maximum and minimum of $\alpha$-$z$-fidelity for several quantum channels, including all quantum channels, unitary quantum channels, and mixed unitary channels. The analysis uses density operators and the closed convex set of quantum channels to provide a geometric interpretation of the \(\alpha\)-\(z\)-fidelity. A key component of this study is the Ostrovsky theorem, which helps determine the extremes of the \(\alpha\)-\(z\)-fidelity, particularly for unitary quantum channels. 
		\par
		Overall, this paper not only provides a systematic exposition of the quantum \(\alpha\)-$z$-fidelity, but also studies its role and properties in practical applications using various mathematical tools and inequalities. This offers new perspectives and theoretical support for the comparison of quantum states and the advancement of quantum information theory.
		\\
		
		\textit{Acknowledgments}: We thank Professor Zhongjin Ruan, Professor Marius Junge, Zhi Yin for very helpful discussions and comments. This work was supported by the Science Foundation of Nanjing University of Posts and Telecommunications (Grant No. NY222025), and the High-level Innovation and Entrepreneurship Talents Program of Jiangsu Province (Grant No. JSS-CBS20220652).
		
		\bibliographystyle{plain}

		\onecolumn\newpage
		\appendix

	\end{document}